\documentclass[onecolumn]{IEEEtran}
\usepackage[cmex10]{amsmath}
\usepackage{amsthm}
\usepackage{amstext}
\usepackage{amssymb}
\usepackage{graphicx}
\usepackage{colortbl}

\makeatletter

\providecommand{\tabularnewline}{\\}

\theoremstyle{definition}
\newtheorem{thm}{Theorem}
\ifx\proof\undefined
\newenvironment{proof}[1][\protect\proofname]{\par
\normalfont\topsep6\p@\@plus6\p@\relax
\trivlist
\itemindent\parindent
\item[\hskip\labelsep
\scshape
#1]\ignorespaces
}{%
\endtrivlist\@endpefalse
}
\providecommand{\proofname}{Proof}
\fi
\theoremstyle{definition}
\newtheorem{lem}{Lemma}
\theoremstyle{remark}

\theoremstyle{definition}
\newtheorem{defn}{Definition}

\usepackage{wasysym}

\@ifundefined{showcaptionsetup}{}{%
 \PassOptionsToPackage{caption=false}{subfig}}
\usepackage{subfig}
\makeatother

\newcommand{\diag}{\mathop{\mathrm{diag}}}

\newcommand{\vv}{\mathbf{v}}
\newcommand{\vu}{\mathbf{u}}

\begin{document}

\title{Implement Blind Interference Alignment over Homogeneous $3$-user $2\times 1$ Broadcast Channel}

\author{Qing F.~Zhou {\em Member, IEEE} , Q. T. Zhang, {\em Fellow,
IEEE} and Francis C. M. Lau, {\em Senior Member, IEEE}
\thanks{Qing F.~Zhou and Q. T. Zhnag are with the Department of Electronic Engineering,
City University of Hong Kong, Kowloon, Hong Kong.
Francis C.~M.~Lau is with the Department of Electronic
and Information Engineering, The Hong Kong Polytechnic University,
Kowloon, Hong Kong. (Email: enqfzhou@ieee.org, wirelessqt@gmail.com,
encmlau@polyu.edu.hk, ).}
\thanks{This work was supported by City University of Hong Kong under xxxxxxx. } }

\maketitle

\begin{abstract}
This paper first studies the homogeneous $3$-user $2\times 1$ broadcast channel (BC) with no CSIT. We show a sufficient condition for it to achieve the optimal $\frac{3}{2}$ degrees of freedom (DoF) by using Blind Interference Alignment (BIA). BIA refers to the interference alignment method without the need of CSIT. It further studies the $2\times 1$ broadcast network in which there are $K\geq 3$ homogeneous single-antenna users, and their coherence time offsets are independently and uniformly distributed. We show that, if $K\geq 11$, the two-antenna transmitter can find, with more than 95\% certainty, three users to form a BIA-feasible $3$-user BC and achieve the optimal $\tfrac{3}{2}$ DoF.
\end{abstract}


\section{Introduction}

In interference channels with generic channel states, the implementation of Interference Alignment (IA) requires the instantaneous and global channel state information at transmitters (CSIT), such as in the seminal $K$-user interference channel \cite{Cadambe2008}. However, in practical transmission systems, the delay introduced by achieving CSIT via feedbacks from the receivers makes the instantaneous and global CSIT unrealistic. Technically speaking, the degree-of-freedom (DoF) region achieved by IA with the assumption of perfect CSIT only serves as an upper bound on achievable DoF.

Be more realistic, IA with delayed CSIT \cite{Maddah2012} and IA with no CSIT \cite{Jafar2009} are of more significance from engineering point of view. Usually reduced DoF is traded for relaxed requirement on CSIT. Consider a $K$-user $L\times 1$ broadcast channel (BC) in which a $L$-antenna transmitter delivers independent data streams to $K$ single-antenna users. When $K=L=2$, $2$ DoF can be achieved by beamforming if CSIT is available, while only $\tfrac{4}{3}$ DoF is achievable if only delayed CSIT is available. Surprisingly, the DoF of $\tfrac{4}{3}$ is also achievable by using IA with no CSIT if the BC contains 4 channel uses and has certain structure on its channel state matrix. By contrast, only $1$ DoF is achievable for the BC if one has no CSIT and uses no IA.

In this paper, we concentrate on IA with no CSIT, which is frequently referred to as Blind Interference Alignment (BIA) \cite{Jafar2009,Gou2011,Jafar2012}. Prior works in the literature only study what the limit of achievable DoF is by using BIA. For instance, it is shown that the optimal DoF of $\tfrac{LK}{L+K-1}$ is achievable for a $K$-user $L\times 1$ BC if it contains finite channel uses and the channel state matrix has certain structure \cite{Jafar2009,Jafar2012}. Whether the optimal DoF is achievable for a general $K$-user $L\times 1$ BC containing infinite channel uses, and how to achieve the optimal DoF are not well studied except our recent work on homogeneous $2$-user $2\times 1$ BC \cite{Zhou2012a}.

\begin{figure}
\begin{centering}
\includegraphics[scale=0.6]{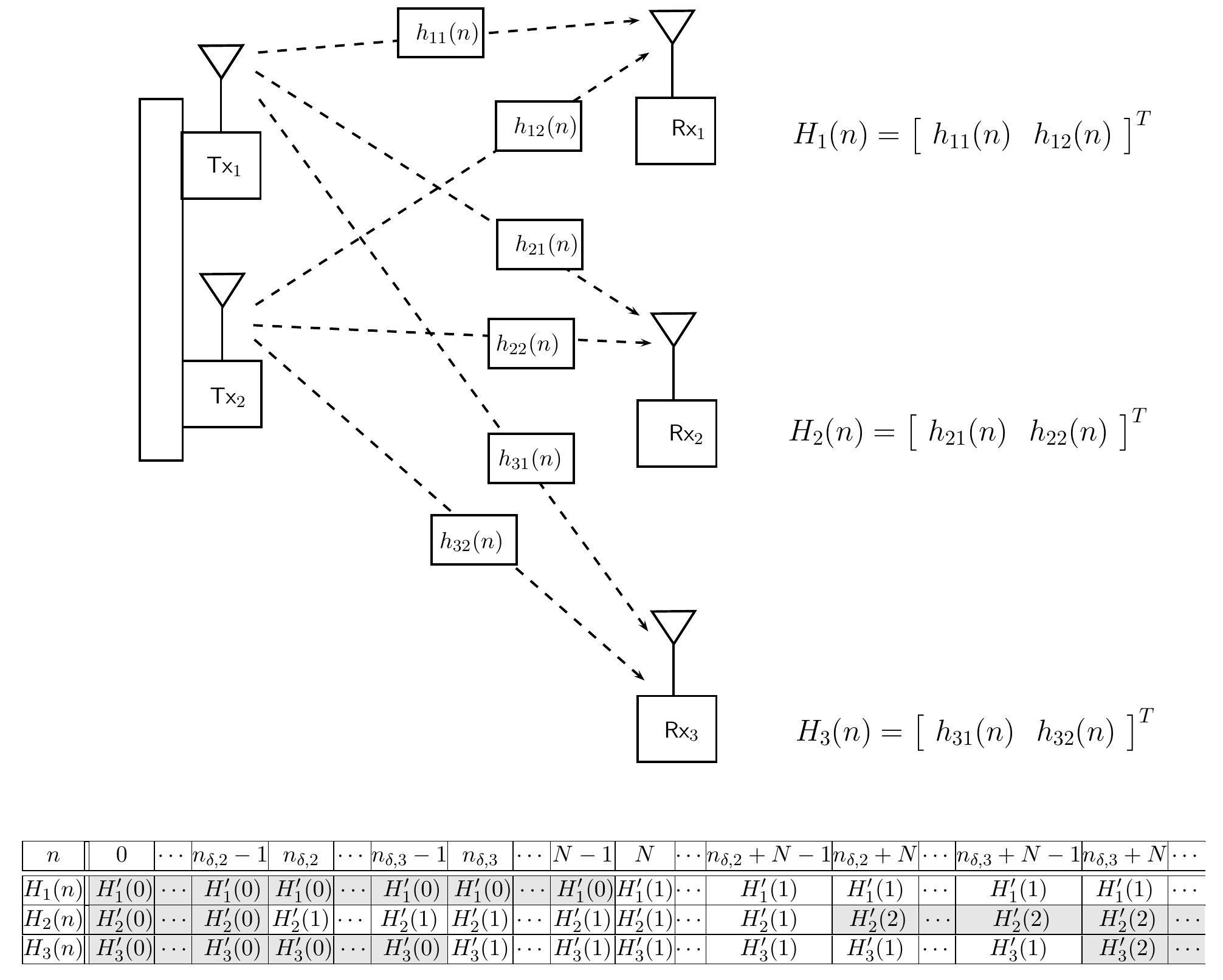}
\par\end{centering}
\caption{System model of a $3$-user $2\times 1$ BC.}
\label{fig:system.model}
\end{figure}

We further examine the condition for a homogeneous $3$-user $2\times 1$ BC to be BIA-feasible, more precisely, to achieve the optimal DoF $\tfrac{3}{2}$ by using BIA. Unlike homogeneous 2-user BC, which can be characterized by a single block offset between the two users \cite{Zhou2012a}, a homogeneous $3$-user BC needs at least two block offsets to parameterize. Due to the doubling of parameter number, it becomes tedious and formidable to analyze the $3$-user BC by using the method in \cite{Zhou2012a}. In particular, the method deals with the explicit decomposition of a $4\times 5N$ channel matrix with $N$ being the coherence time. To reduce the complexity, we propose a new analysis method, which simplifies the decomposition of the $4\times 5N$ channel matrix into the decomposition of a $12\times 12$ matrix regardless of the value of $N$. By using this new method, we identify a sufficient condition on block offsets for a homogeneous $3$-user BC to be BIA-feasible. If there are $K\geq 11$ homogeneous users in a $2\times 1$ broadcast network, we further show that it is highly possible that the two-antenna transmitter can find, from the $K$ users, three users to form a BIA-feasible $3$-user BC.

\section{System model}
A $3$-user $2\times 1$ broadcast channel (BC) consists of a two-antenna transmitter and three single-antenna users. In this paper, we consider a $3$-user $2\times 1$ BC in which the three users undertake independent block fading with the same coherence time, referred to as homogeneous $3$-user $2\times 1$ BC. As shown in Fig.~\ref{fig:system.model}, let $h_{ij}$ denote the channel coefficient of the link connecting transmit antenna $\mathsf{Tx}_{j}$, $j\in\{1,2\}$,
to user $\mathsf{Rx}_{i}$, $i\in\{1,2,3\}$. Further define coefficient vector $H_{i}(n)\triangleq [h_{i1}(n), h_{i2}(n)]^{T}$ for $\mathsf{Rx}_{i}$, where $n\geq 0$ is the discrete-time index. Represent the coherence time for $H_{i}(n)$ by $N$, which is the same for all users,  and the initial time offset for $H_{i}(n)$ by $n_{\delta,i}$. Without loss of generality, let $n_{\delta,1}=0$, and $0\leq n_{\delta,i}<N$ for $i\in\{2,3\}$.
So, as illustrated in Fig.~\ref{fig:system.model}, we can let $H'_i(a_i)$ denote the coefficient vector for user $\mathsf{Rx}_i$ at its $a_i$th coherence block. In specific, for user $\mathsf{Rx}_1$ we have $H_{1}(a_1N+b_{1})=H'_1(a_1)$ for all $a_1 \geq 0$ and $0\leq b_{1}<N$; for users $\mathsf{Rx}_2$ and $\mathsf{Rx}_3$, we have $H_{i}(b_{i})=H'_i(0)$ for $0\leq b_{i}\leq n_{\delta,i}-1$, and $H_{i}(a_i N+n_{\delta,i}+b_{i})=H'_i(a_i+1)$
for $a_i \geq 0$ and $0\leq b_i<N$, where $i\in\{ 2,3\}$.
%
%
%
%

\subsection{Interference alignment}
It is known that the optimal achievable DoF, also known as multiplexing gain, is $\frac{3}{2}$ for the $3$-user $2\times 1$ BC with no CSIT \cite{Gou2011}. As an illustration, Fig.~\ref{fig:IA.X.chnl} shows how interference alignment helps to achieve the optimal DoF.
In this figure, we consider four arbitrary time slots $n_{1}< n_{2}< n_{3} < n_{4}$, which are not necessarily consecutive. Their channel coefficients are given by
\begin{center}
\begin{tabular}{|c|c|c|c|}
\hline
$H_{1}(n_{1})$ & $H_{1}(n_{2})$ & $H_{1}(n_{3})$ & $H_{1}(n_{4})$\tabularnewline
\hline
$H_{2}(n_{1})$ & $H_{2}(n_{2})$ & $H_{2}(n_{3})$ & $H_{2}(n_{4})$\tabularnewline
\hline
$H_{3}(n_{1})$ & $H_{3}(n_{2})$ & $H_{3}(n_{3})$ & $H_{3}(n_{4})$\tabularnewline
\hline
\end{tabular}\quad.
\par\end{center}
Since each transmit antenna delivers one symbol on each time slot, the four-time-slot channel block is also referred to as a \emph{4-symbol channel block}.
Over the block, we denote channel coefficient matrix from transmit antenna $\mathsf{Tx}_{j}$ to user
$\mathsf{Rx}_{i}$ as $H_{ij}=\mathsf{diag}[h_{ij}(n_{1}),h_{ij}(n_{2}),h_{ij}(n_{3}),h_{ij}(n_{4})]$.
Suppose $\mathbf{v}_{1},\mathbf{v}_{2},\mathbf{v}_{3}\in\mathcal{C}^{4\times 1}$ are the signaling
vectors for $\mathsf{Tx}_{1}$; $\mathbf{u}_{1},\mathbf{u}_{2},\mathbf{u}_{3}\in\mathcal{C}^{4\times1}$
are the signaling vectors for $\mathsf{Tx}_{2}$. At user $\mathsf{Rx}_i$, the received signal vector $\mathbf{y}_i\in\mathcal{C}^{4\times 1}$ is
\begin{equation}
\mathbf{y}_i = H_{i1} [ \mathbf{v}_1,\mathbf{v}_2,\mathbf{v}_3 ] \mathbf{s}_1 + H_{i2}[ \mathbf{u}_1,\mathbf{u}_2,\mathbf{u}_3 ] \mathbf{s}_2 + \mathbf{z}_j
\end{equation}
where $\mathbf{s}_j = [s_{j1},s_{j2},s_{j3}]^T \in \mathcal{C}^{3\times 1}$ represents three independent data streams from transmit antenna $\mathsf{Tx}_j$, $\mathbf{z}_i \in \mathcal{C}^{4\times 1}$ is the AWGN vector at user $\mathsf{Rx}_i$.
\begin{figure}
\begin{centering}
\includegraphics[scale=0.7]{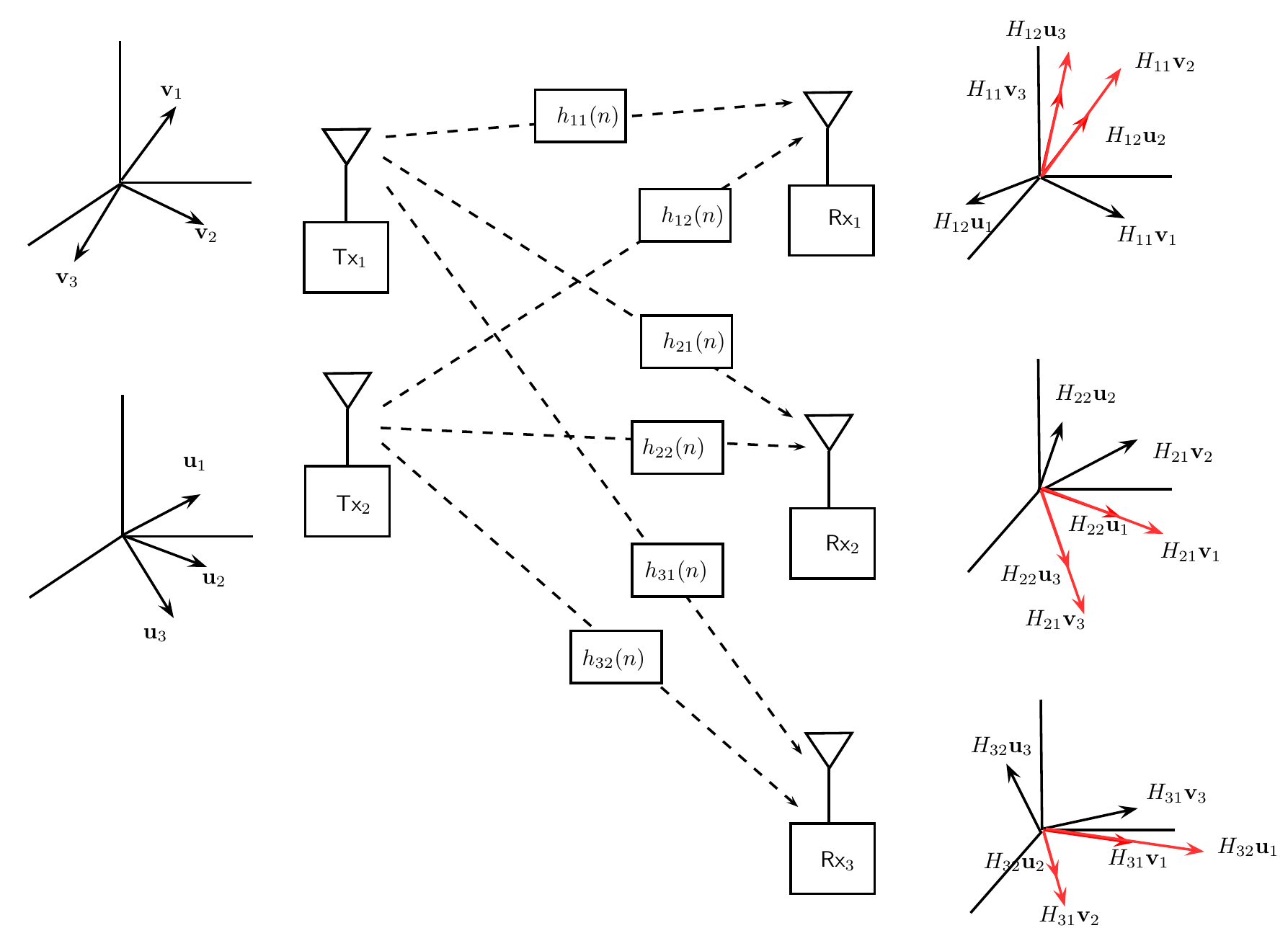}
\par\end{centering}
\caption{Interference alignment of a $3$-user $2\times 1$ BC. Note that the symbol streams $s_{ij}$ are omitted for simplicity.}
\label{fig:IA.X.chnl}
\end{figure}

To achieve the optimal DoF of $\tfrac{4}{3}$,
the interference alignment implementation shown in Fig.~\ref{fig:IA.X.chnl} requires
\begin{equation}
\begin{cases}
H_{11}\mathbf{v}_{2}\rightarrow H_{12}\mathbf{u}_{2}\\
H_{11}\mathbf{v}_{3}\rightarrow H_{12}\mathbf{u}_{3}
\end{cases}\label{eq:IA.cond.01}
\end{equation}
and %
\begin{equation}
\begin{cases}
H_{21}\mathbf{v}_{1}\rightarrow H_{22}\mathbf{u}_{1}\\
H_{21}\mathbf{v}_{3}\rightarrow H_{22}\mathbf{u}_{3}
\end{cases}\label{eq:IA.cond.02}
\end{equation}
and %
\begin{equation}
\begin{cases}
H_{31}\mathbf{v}_{1}\rightarrow H_{32}\mathbf{u}_{1}\\
H_{31}\mathbf{v}_{2}\rightarrow H_{32}\mathbf{u}_{2}
\end{cases}\label{eq:IA.cond.03}
\end{equation}
where $\mathbf{x}\rightarrow\mathbf{y}$ means that $\mathbf{x}$ aligns
with $\mathbf{y}$, that is, $\mathbf{x}=c\mathbf{y}$ for a nonzero
scalar $c$.
%
In this implementation, $\mathsf{Rx}_i$, $i\in\{1,2,3\}$, decodes the symbols delivered by $\mathbf{v}_i$ and $\mathbf{u}_i$, i.e., $s_{1i}$ and $s_{2i}$ from $\mathsf{Tx}_1$ and $\mathsf{Tx}_2$, respectively. In total, six symbols are delivered by four channel uses, so the DoF $\tfrac{6}{4}=\tfrac{3}{2}$ is achieved. The alignment criteria given above can be rewritten as \eqref{eq:sig.vec.cond}.
\begin{figure*}
\begin{equation}
\begin{cases}
\vv_{1} & \rightarrow\diag\left(\frac{h_{22}(n_{1})}{h_{21}(n_{1})},\frac{h_{22}(n_{2})}{h_{21}(n_{2})},\frac{h_{22}(n_{3})}{h_{21}(n_{3})},\frac{h_{22}(n_{4})}{h_{21}(n_{4})}\right)\vu_{1}\rightarrow\diag\left(\frac{h_{32}(n_{1})}{h_{31}(n_{1})},\frac{h_{32}(n_{2})}{h_{31}(n_{2})},\frac{h_{32}(n_{3})}{h_{31}(n_{3})},\frac{h_{32}(n_{4})}{h_{31}(n_{4})}\right)\vu_{1}\\
\vv_{2} & \rightarrow\diag\left(\frac{h_{12}(n_{1})}{h_{11}(n_{1})},\frac{h_{12}(n_{2})}{h_{11}(n_{2})},\frac{h_{12}(n_{3})}{h_{11}(n_{3})},\frac{h_{12}(n_{4})}{h_{11}(n_{4})}\right)\vu_{2}\rightarrow\diag\left(\frac{h_{32}(n_{1})}{h_{31}(n_{1})},\frac{h_{32}(n_{2})}{h_{31}(n_{2})},\frac{h_{32}(n_{3})}{h_{31}(n_{3})},\frac{h_{32}(n_{4})}{h_{31}(n_{4})}\right)\vu_{2}\\
\vv_{3} & \rightarrow\diag\left(\frac{h_{12}(n_{1})}{h_{11}(n_{1})},\frac{h_{12}(n_{2})}{h_{11}(n_{2})},\frac{h_{12}(n_{3})}{h_{11}(n_{3})},\frac{h_{12}(n_{4})}{h_{11}(n_{4})}\right)\vu_{3}\rightarrow\diag\left(\frac{h_{22}(n_{1})}{h_{21}(n_{1})},\frac{h_{22}(n_{2})}{h_{21}(n_{2})},\frac{h_{22}(n_{3})}{h_{21}(n_{3})},\frac{h_{22}(n_{4})}{h_{21}(n_{4})}\right)\vu_{3}
\end{cases}. \label{eq:sig.vec.cond}
\end{equation}
\end{figure*}

\subsection{BIA-feasible channel block}

In this part, we investigate the way to meet the interference alignment condition \eqref{eq:sig.vec.cond}
when we have no CSIT but the knowledge of the coherence time $N$ and
the offsets $n_{\delta,i}$, $i\in\{1,2,3\}$ at the transmitter. This type of interference alignment requiring no CSIT is referred to as BIA \cite{Gou2011}.

According to the results in \cite{Gou2011}, when a 4-symbol channel block presents the following channel coefficient pattern
\begin{center}
\begin{tabular}{|c|c|c|c|}
\hline
$H_{1}'(\beta)$ & \cellcolor[gray]{0.9}$H_{1}'(\alpha)$ & $H_{1}'(\beta)$ & $H_{1}'(\beta)$\tabularnewline
\hline
$H_{2}'(\gamma)$ & $H_{2}'(\gamma)$ & \cellcolor[gray]{0.9}$H_{2}'(\psi)$ & $H_{2}'(\gamma)$\tabularnewline
\hline
$H_{3}'(\rho)$ & $H_{3}'(\rho)$ & $H_{3}'(\rho)$ & \cellcolor[gray]{0.9}$H_{3}'(\pi)$\tabularnewline
\hline
\end{tabular}\quad,
\par\end{center}
it is feasible for BIA to achieve the optimal $\tfrac{3}{2}$. For simplicity of presentation, we say the block is \emph{BIA-feasible}. In any homogeneous $3$-user $2\times 1$ BC, one cannot find four time slots to form a BIA-feasible block with the channel pattern above. Fortunately, the channel pattern does not form the necessary condition for a 4-symbol channel block to be BIA-feasible, other channel patterns can also form BIA-feasible 4-symbol channel blocks. For example, a 4-symbol channel block is BIA-feasible if it has the following channel pattern:
\begin{center}
\begin{tabular}{|c|c|c|c|}
\hline
\cellcolor[gray]{0.9}$H_{1}'(\alpha)$ & \cellcolor[gray]{0.9}$H_{1}'(\alpha)$ & $H_{1}'(\beta)$ & $H_{1}'(\beta)$\tabularnewline
\hline
\cellcolor[gray]{0.9}$H_{2}'(\gamma)$ & $H_{2}'(\psi)$ & $H_{2}'(\psi)$ & $H_{2}'(\psi)$\tabularnewline
\hline
\cellcolor[gray]{0.9}$H_{3}'(\rho)$ & \cellcolor[gray]{0.9}$H_{3}'(\rho)$ & \cellcolor[gray]{0.9}$H_{3}'(\rho)$ & $H_{3}'(\pi)$\tabularnewline
\hline
\end{tabular}
\par\end{center}
To see its BIA-feasibility, we substitute the channel state into \eqref{eq:sig.vec.cond} and get the alignment conditions given by \eqref{eq:ex.pattern}.
\begin{figure*}
\begin{equation}
\begin{cases}
\vv_{1} & \rightarrow\diag\left(\frac{h_{22}'(\gamma)}{h_{21}'(\gamma)},\frac{h_{22}'(\psi)}{h_{21}'(\psi)},\frac{h_{22}'(\psi)}{h_{21}'(\psi)},\frac{h_{22}'(\psi)}{h_{21}'(\psi)}\right)\vu_{1}\rightarrow\diag\left(\frac{h_{32}'(\rho)}{h_{31}'(\rho)},\frac{h_{32}'(\rho)}{h_{31}'(\rho)},\frac{h_{32}'(\rho)}{h_{31}'(\rho)},\frac{h_{32}'(\pi)}{h_{31}'(\pi)}\right)\vu_{1}\\
\vv_{2} & \rightarrow\diag\left(\frac{h_{12}'(\alpha)}{h_{11}'(\alpha)},\frac{h_{12}'(\alpha)}{h_{11}'(\alpha)},\frac{h_{12}'(\beta)}{h_{11}'(\beta)},\frac{h_{12}'(\beta)}{h_{11}'(\beta)}\right)\vu_{2}\rightarrow\diag\left(\frac{h_{32}'(\rho)}{h_{31}'(\rho)},\frac{h_{32}'(\rho)}{h_{31}'(\rho)},\frac{h_{32}'(\rho)}{h_{31}'(\rho)},\frac{h_{32}'(\pi)}{h_{31}'(\pi)}\right)\vu_{2}\\
\vv_{3} & \rightarrow\diag\left(\frac{h_{12}'(\alpha)}{h_{11}'(\alpha)},\frac{h_{12}'(\alpha)}{h_{11}'(\alpha)},\frac{h_{12}'(\beta)}{h_{11}'(\beta)},\frac{h_{12}'(\beta)}{h_{11}'(\beta)}\right)\vu_{3}\rightarrow\diag\left(\frac{h_{22}'(\gamma)}{h_{21}'(\gamma)},\frac{h_{22}'(\psi)}{h_{21}'(\psi)},\frac{h_{22}'(\psi)}{h_{21}'(\psi)},\frac{h_{22}'(\psi)}{h_{21}'(\psi)}\right)\vu_{3}
\end{cases}.\label{eq:ex.pattern}
\end{equation}
\end{figure*}
The conditions would be satisfied if we choose $\vv_{1}=\vu_{1}=[0,1,1,0]^{T}$,
$\vv_{2}=\vu_{2}=[1,1,0,0]^{T}$ and $\vv_{3}=\vu_{3}=[0,0,1,1]^{T}$,
and thus the channel block achieves interference alignment with
no CSIT, i.e., BIA-feasibility.

Over the exemplary 4-symbol channel block shown above, the first two symbols of user $\mathsf{Rx}_1$ have the same channel state while the last two symbols have also the same channel state but different from the first two. To reflect the distribution of these two channel states over time, we say user $\mathsf{Rx}_1$ has the channel structure $(2,2)$ over the 4-symbol channel block. Similarly, user $\mathsf{Rx}_2$ has the channel structure $(1,3)$, and user $\mathsf{Rx}_3$ has the channel structure $(3,1)$. We say the 4-symbol channel block has the channel pattern $\{(2,2),(1,3),(3,1)\}$. It is easy to verify that a 4-symbol channel block is BIA-feasible if its channel pattern is any permutation of $(1,3)$, $(2,2)$ and $(3,1)$, such as $\{(1,3),(3,1),(2,2)\}$. The number of such permutations is six, therefore, one can find six kinds of BIA-feasible 4-symbol channel block from homogeneous $3$-user $2\times 1$ BCs.

\begin{figure*}
\includegraphics[scale=0.88]{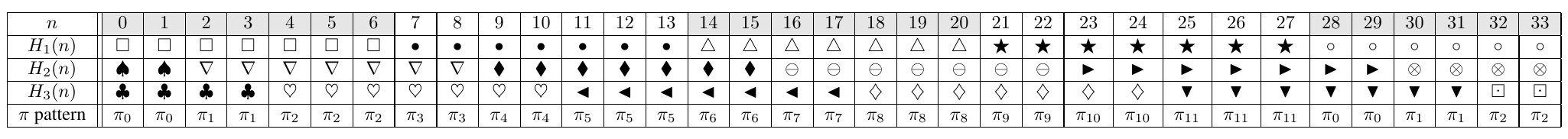}
\caption{\label{table:1}A homogeneous BC with $N=7$,
$n_{\delta,2}=2$ and $n_{\delta,3}=4$.}
\end{figure*}
We say a homogeneous $3$-user $2\times 1$ BC with coherence time $N$ and time offsets $n_{\delta ,i}$ is \emph{BIA-feasible if the broadcast channel can be fully decomposed into BIA-feasible 4-symbol channel blocks}.
We give a simple demonstration to show this argument. Consider a homogeneous BC with $N=7$, $n_{\delta,2}=2$ and $n_{\delta,3}=4$; its channel coefficients over time are presented in Fig.~\ref{table:1}, in which the same marks represent the same channel state.
The channel block from $n=6$ to $n=33$, which contains $4N=28$ consecutive time slots, can be decomposed
into $N=7$ BIA-feasible 4-symbol channel blocks as presented in Fig.~\ref{table:2}.
\begin{figure*}
\center
\includegraphics[scale=0.85]{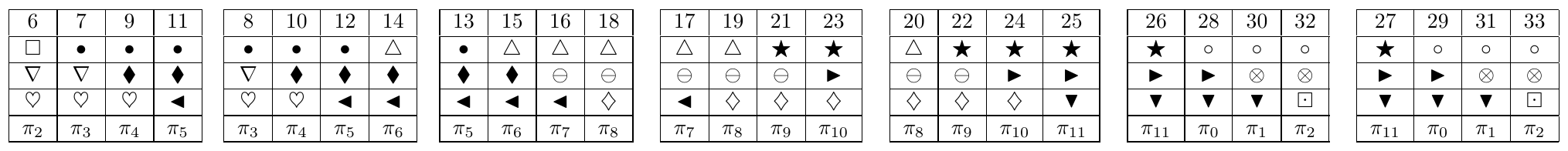}
\caption{\label{table:2}The implementation of BIA for the homogeneous
BC shown in Fig. \ref{table:1}.}
\end{figure*}
For instance, the time slots $n_1=8$, $n_2=10$, $n_3=12$ and $n_4=14$ form a BIA-feasible channel block with the channel pattern $\{(3,1),(1,3),(2,2)\}$. Since the channel repeats the same pattern every $28$ consecutive symbols afterwards, the same decomposition
can be achieved repetitively. Therefore, this channel can obtain the optimal $\tfrac{3}{2}$ DoF by using BIA. Note that the first three time slots are negligible or can be utilized for control messages.

\section{Sufficient condition for BIA-feasibility}

In the example above, we note that, from $n=n_1$ to the next time $n=n_1+1$, channel state $\{H_1(n),H_2(n),H_3(n)\}$ may or may not remain unchanged. For instance, the channel state remains the same from $n=0$ to $n=1$ in Fig.~\ref{table:1}, but varies from $n=1$ to $n=2$. We collect consecutive time slots having identical channel states into a group, and assign the time slots in the same group with the same label ($\pi$). In the example, the $0$th group contains $n=0,1$, and is given the label $\pi_0$; afterwards, the $t$th group is given the label $\pi_{t}$ until $n=4N-1=27$th time slot has been considered. For every consecutive block of $4N$ time slots, we repeat the grouping process, and start labeling groups from $0$ again.

Next we generalize the grouping and labeling idea to an arbitrary integer $N$. Without loss of generality, we assume $n_{\delta,2}\leq n_{\delta,3}$. Let $s_t$ be the number of $\pi_t$'s in the $t$th group.
Then we can easily verify that
$s_0 = n_{\delta,2}$, $s_1 = n_{\delta,3} -n_{\delta,2}$, $s_2 = N - n_{\delta,3}$ and
\begin{equation}\label{eq:s_t}
s_t = s_{t'}, \text{  if } t\equiv t'\mod{3},
\end{equation}
where  $t'\in\{0,1,2\}$. The constant term `$3$' in the equation above reflects the number of users. So, the example given by Fig.~\ref{table:1} has $s_0=2$, $s_1=2$ and $s_2=3$. By this definition, the channel block from $n=0$ to $n=N-1$ can be divided into 3 groups, which have sizes $s_0$, $s_1$ and $s_2$, respectively. The ``super" channel block from $n=0$ to $n=4N-1$, a block containing four coherence blocks, can be divided into $3\times 4 = 12$ groups, which have sizes $s_t$, $t\in\{0,1,\cdots,11\}$.

It is clear that a homogeneous $3$-user $2\times 1$ BC can be fully characterized by $s_t$. Based on $s_t$, we can visualize the BC by constructing a \emph{pattern diagram}. In specific, we align, along the $t$th column of the pattern diagram, the $\pi_t$'s in the $t$th group. As an example, Fig.~\ref{fig:ex.N4} shows the pattern diagram of the homogeneous BC given by Fig.~\ref{table:1}.
\begin{figure}
\begin{centering}
\includegraphics[scale=0.8]{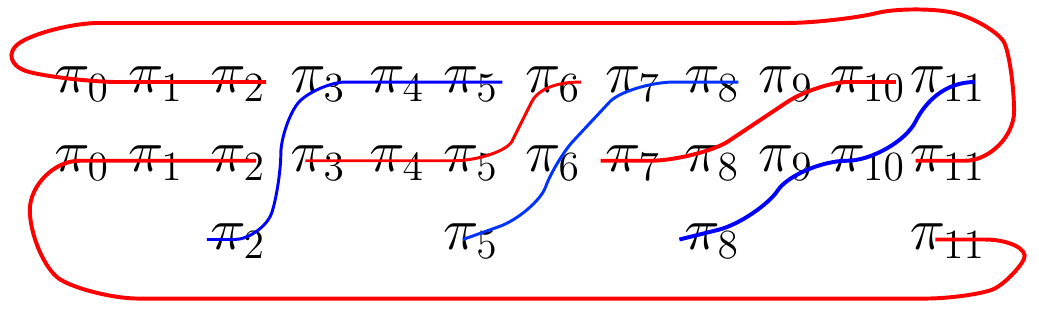}
\par\end{centering}
\caption{\label{fig:ex.N4}The pattern diagram for the homogeneous BC with $N=7$, $n_{\delta,2}=2$ and $n_{\delta,3}=4$ shown in Fig.~\ref{table:1}. As illustrated, this pattern diagram can be fully decomposed by BIA-feasible 4-tuples.}
\end{figure}

It is easy to verify, from Fig.~\ref{table:2} as an example, that any 4-symbol channel block formed by $(\pi_t,\pi_{t+1},\pi_{t+2},\pi_{t+3})$ is BIA-feasible, where the subscript of $\pi$ takes value from $\mathcal{Z}_{12}$, the integer ring on the base of $12$. Given a pattern diagram, let us connect the elements $\pi_t$ by threads, and each thread connects a 4-tuple $(\pi_t,\pi_{t+1},\pi_{t+2},\pi_{t+3})$. Then we have the following definition.
\begin{defn}
We say a pattern diagram is able to be \emph{completely decomposed} into $N$ 4-tuple $(\pi_t,\pi_{t+1},\pi_{t+2},\pi_{t+3})$ if each element $\pi_t$ is connected by one and only one thread.
\end{defn}
Fig.~\ref{fig:ex.N4} shows an example of completely decomposable channel pattern. Then, we have the following lemma.
\begin{lem}
A homogeneous $3$-user $2\times 1$ BC is BIA-feasible if its pattern diagram is completely decomposable.
\end{lem}
By using the lemma, we can find a sufficient condition for a homogeneous $3$-user $2\times 1$ to be BIA-feasible in terms of $s_i$, which is given by the following theorem.
\begin{thm}\label{thm:feasible.region.s}
A homogeneous $3$-user $2\times 1$ BC is BIA-feasible if the following inequality is satisfied:
\begin{equation}\label{eq:BIA.cond}
\sum_{i=0}^2 s_i \leq 4\min(s_0,s_1,s_2).
\end{equation}
\end{thm}
\begin{proof}
\begin{figure*}
\begin{centering}
\includegraphics[scale=0.9]{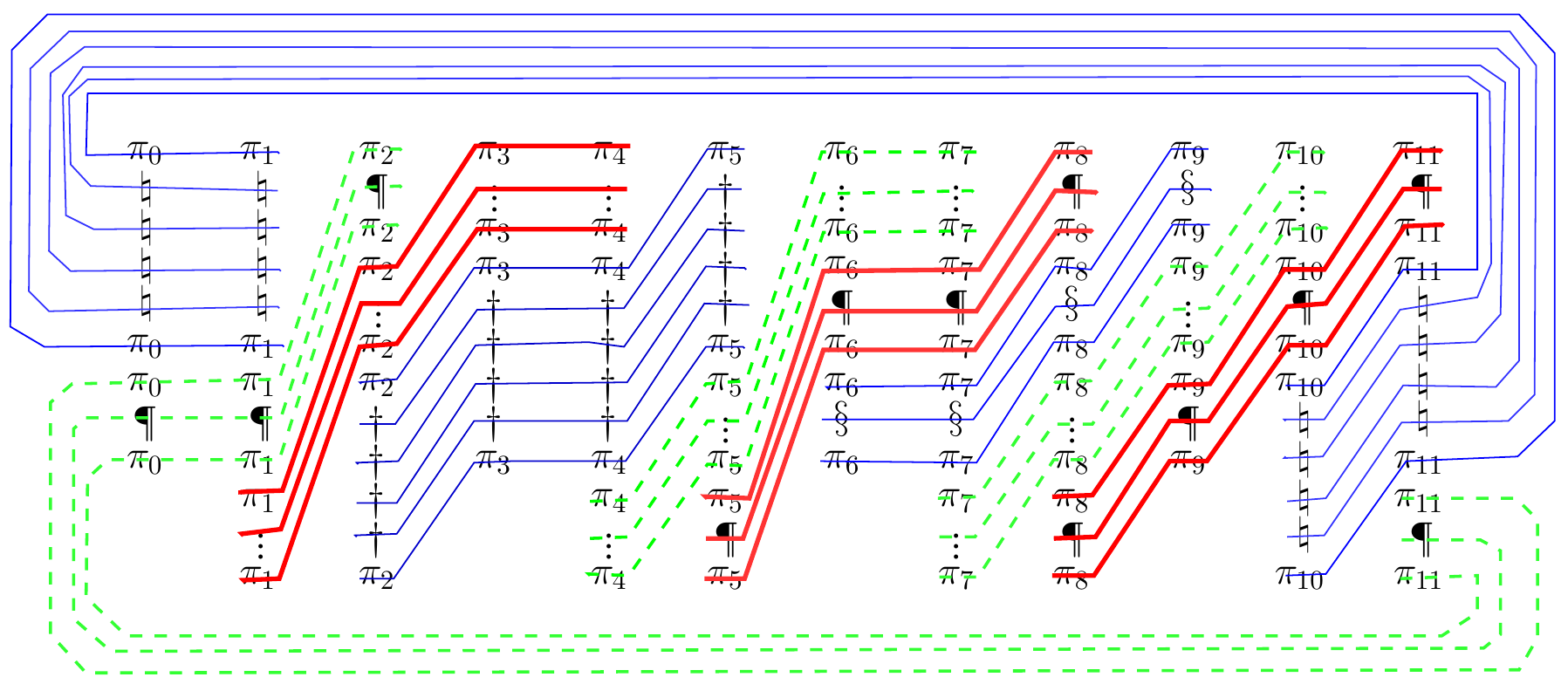}
\caption{Complete decomposition of the pattern diagram of a homogeneous $3$-user $2\times 1$ BC which meets the sufficient condition given by \eqref{eq:BIA.cond}.}
\label{figure:proof.thm01}
\end{centering}
\end{figure*}
We prove the theorem by showing a constructive algorithm which can completely decomposes a pattern diagram, if it satisfies \eqref{eq:BIA.cond}, into $N$ 4-tuple $(\pi_t,\pi_{t+1},\pi_{t+2},\pi_{t+3})$.

Without loss of generality, we assume $s_0 = \min (s_0,s_1,s_2)$. Fig.~\ref{figure:proof.thm01} shows the pattern diagram of a homogeneous $3$-user $2\times 1$ BC which meets the sufficient condition \eqref{eq:BIA.cond}. This figure also shows the threads which are designed to completely decompose the pattern diagram. Let $l_t$ be the number of threads beginning from $\pi_t$, $t\in \mathcal{Z}_{12}$. To elaborate the decomposition, we start with the threads beginning with $\pi_1$. We use $l_1 = s_1 - s_0$ such threads, and denote $\vdots$ as the $\pi_t$'s connected by this number of threads. The values of other $l_t$'s and the notations associated with them are listed in Table.~\ref{table:threads}. From $\sum_{i=0}^2 s_i \leq 4s_0$, we can prove $l_t\geq 0$ for all $t\in \mathcal{Z}_{12}$.
\begin{table*}
\begin{center}
\begin{tabular}{|c||c|c|c|c|c|c|c|c|c|c|c|c|}
\hline
 & $l_{0}$ & $l_{1}$ & $l_{2}$ & $l_{3}$ & $l_{4}$ & $l_{5}$ & $l_{6}$ & $l_{7}$ & $l_{8}$ & $l_{9}$ & $l_{10}$ & $l_{11}$\tabularnewline
\hline
Size & 0 & $s_{1}-s_{0}$ & $2s_{0}-s_{1}$ & 0 & $s_{1}-s_{0}$ & $s_{2}-s_{0}$ & $3s_{0}-s_{1}-s_{2}$ & $s_{1}-s_{0}$ & $s_{2}-s_{0}$ & 0 & $2s_{0}-s_{2}$ & $s_{2}-s_{0}$\tabularnewline
\hline
Notations &  & $\vdots$  & $\dagger$ &  & $\vdots$ & $\P$ & $\S$ & $\vdots$ & $\P$ &  & $\natural$ & $\P$ \tabularnewline
\hline
\end{tabular}
\end{center}
\caption{The size of the threads beginning from $\pi_t$ and the notation.}\label{table:threads}
\end{table*}

From Table.~\ref{table:threads}, we can easily verify that the equation
\begin{equation}
s_t = l_{t-3} + l_{t-2} + l_{t-1} + l_{t}
\end{equation}
is valid for all $t\in\mathcal{Z}_{12}$. For instance, when $t=7$, we have $s_7 = s_1$ according to \eqref{eq:s_t}, and $l_4 + l_5 + l_6 + l_7 = s_1-s_0 + s_2-s_0 + 3s_0-s_1-s_2 + s_1-s_0 = s_1$. This proves that the decomposition given by Table.~\ref{table:threads}, also shown in Fig.~\ref{figure:proof.thm01}, is a valid complete decomposition. So, the theorem is proved.
\end{proof}

\section{Impact of coherence time $N$ and user number $K$}

Physically, given that $n_{\delta,1}=0$ is fixed, it is justified to model $n_{\delta,i}$, $i\in\{2,3\}$, as independent random variables uniformly distributed over $0\leq n_{\delta,i}\leq N-1$.

\subsection{Impact of coherence time $N$}
Given $N$ and $n_{\delta,1}=0$, there are $N^2$ pairs of $(n_{\delta,2},n_{\delta,3})$. However, not all of them are able to satisfy the sufficient condition \eqref{eq:BIA.cond} to ensure a BIA-feasible BC. The following lemma counts the number of pairs which meet the sufficient condition.
\begin{lem}
Given a homogeneous $3$-user $2\times 1$ BC with coherence time $N$ and $n_{\delta,1}=0$, let $f(N,3)$ denote the number of $(n_{\delta,2},n_{\delta,3})$ pairs which can generate BIA-feasible channels. Then we have
\begin{equation}
f(N,3)
  =  2\sum_{s_0 = \lceil \tfrac{N}{4}\rceil }^{\lfloor\tfrac{N}{3} \rfloor}
\left( 6( \lfloor \tfrac{N-s_0}{2} \rfloor  -s_0) + 3 - 3\left\lfloor \lfloor\tfrac{N-s_0}{2} \rfloor \tfrac{2}{N-s_0}\right \rfloor  \right) + 2\left\lfloor \lfloor\tfrac{N}{3} \rfloor \tfrac{3}{N}\right\rfloor,
\end{equation}
where $\lfloor \cdot \rfloor$ represents the floor function, $\lceil \cdot \rceil $ represents the ceiling function.
\end{lem}
\begin{IEEEproof}
Let $s_0\leq s_1\leq s_2$. Then according to $\sum_{i=0}^2 s_i = N$ and the sufficient condition $\sum_{i=0}^2 s_i \leq 4s_0$, we have the constraint $3s_0\leq N\leq 4s_0$. As $s_0$ must be an integer, we rewrite the constaint on $s_0$ as $\lceil \tfrac{N}{4} \rceil \leq s_0 \leq \lfloor \tfrac{N}{3}\rfloor$. Given any $s_0$ satisfying the constraint, $s_1$ can take any value from $s_0\leq s_1\leq \lfloor \tfrac{N-s_0}{2}\rfloor$, resulting in $\lfloor \tfrac{N-s_0}{2}\rfloor - s_0 +1$ choices. Once $s_0$ and $s_1$ are chosen, then $s_2$ is determined by $s_2 = N-s_0-s_1$.

Suppose $n_{\delta,2}<n_{\delta,3}$. For any triplet $(s_0,s_1,s_2)$, we can assign them to $n_{\delta,2}$ and $n_{\delta,3}$ using
$n_{\delta,2} = s_{i1}$, $n_{\delta,3} - n_{\delta,2} = s_{i2}$ and $N-n_{\delta,3} = s_{i3}$, where $(i1,i2,i3)$ is  any permutation of $(0,1,2)$. Each triplet gives 6 possible pairs of $(n_{\delta,2},n_{\delta,3})$ if $s_0$, $s_1$ and $s_2$ are all different; 3 pairs if two of $s_i$'s are equal; 1 pair if all of them are equal. So, given $n_{\delta,2}<n_{\delta,3}$, the number of pairs of $(n_{\delta,2},n_{\delta,3})$ which meet the sufficient condition is given by
\begin{align}
g(N,3) &=\sum_{s_0 = \lceil \tfrac{N}{4}\rceil }^{\lfloor\tfrac{N}{3} \rfloor}
\left( 6( \lfloor \tfrac{N-s_0}{2} \rfloor  -s_0) + 3 - 3\left\lfloor \lfloor\tfrac{N-s_0}{2} \rfloor \tfrac{2}{N-s_0}\right \rfloor  \right) + \left\lfloor \lfloor\tfrac{N}{3} \rfloor \tfrac{3}{N}\right\rfloor,
\end{align}
where the constant term `3' in the summation counts the number of pairs when $s_1=s_0$; the term `$3\left\lfloor \lfloor\tfrac{N-s_0}{2} \rfloor \tfrac{2}{N-s_0}\right \rfloor$' accounts for the cases when $\tfrac{N-s_0}{2} $ is an integer and thus  $s_1=s_2$; and the term `$\left\lfloor \lfloor\tfrac{N}{3} \rfloor \tfrac{3}{N}\right\rfloor$' accounts for the case when $s_0 = \tfrac{N}{3}$ and $s_2=s_1=s_0$.

By realizing the same number of pairs are available if $n_{\delta,2}>n_{\delta,3}$, we get
\begin{equation}
f(N,3) = 2g(N,3),
\end{equation}
which proves the theorem.
\end{IEEEproof}

The lemma above clearly shows that the number of BIA-feasible BCs is determined by the coherence time $N$. Based on the assumption that $n_{\delta,i}$s are independently and uniformly distributed, each pair of $(n_{\delta,2},n_{\delta,3})$ would occur with the probability of $\tfrac{1}{N^2}$. By using the lemma above, we can derive the probability that a homogeneous $3$-user $2\times 1$ BC is BIA-feasible as follows.
\begin{thm}
Given a homogeneous $3$-user $2\times 1$ BC with coherence time $N$, the probability that it is BIA-feasible is given by
\begin{equation}
P(N,3) = \frac{f(N,3)}{N^2}.
\end{equation}
\end{thm}

\subsection{Impact of user number $K$}
Now, consider a $2\times 1$ broadcast network with $K\geq4$ mobile users experiencing homogeneous block fading. Suppose
the fading block offset for $\mathsf{Rx}_{i}$ ($2\leq i\leq K$)
is independently and uniformly distributed over $0\leq n_{\delta,i}\leq N-1$, with respect to $n_{\delta,1}=0$. Now, we examine the probability
that the $2$-antenna transmitter can find, among the $K$ users, three users to form
a BIA-feasible $3$-user $2\times 1$ BC.

We first give a lemma showing the BIA-feasible condtion on $(n_{\delta,1},n_{\delta,2},n_{\delta,3})$.
\begin{lem}
\label{lem:3.user.no.benchmark}Three users form a BIA-feasible BC if
$|n_{\delta,i}-n_{\delta,j}|\geq\left\lceil \frac{N}{4}\right\rceil $
holds for any pair of $i\neq j$.
\end{lem}
\begin{figure}
\begin{centering}
\includegraphics[scale=1.5]{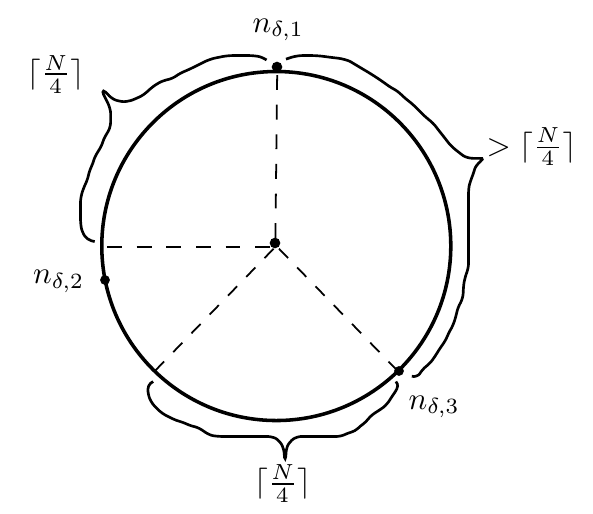}
\par\end{centering}

\caption{\label{fig:3.user.ring}A BIA-feasible 3-tuple $(n_{\delta,1},n_{\delta,2},n_{\delta,3})$.}
\end{figure}

\begin{IEEEproof}
Illustratively, we can visualize the condition by Fig. \ref{fig:3.user.ring},
in which any pair of $n_{\delta,i}$ and $n_{\delta,j}$ is separated
by at least $\lceil\frac{N}{4}\rceil$. The lemma can be easily proved
by setting one user as the benchmark, say $n_{\delta,1}=0$, and then
applying Theorem \ref{thm:feasible.region.s}.
\end{IEEEproof}
Next we show another lemma which is about to be used in the following analysis.
\begin{lem}
Suppose there are $n$ labeled boxes, and $\Theta$ labeled balls.
Given $\mu \leq \min\{n,\Theta\}$ boxes, the number of ways to put the balls into the boxes
such that the $\mu$ boxes are not empty is given by
\begin{equation}
\gamma(n,\Theta,\mu)=\sum_{k=\mu}^{\Theta}\binom{\Theta}{k}\mu!S(k,\mu)(n-\mu)^{\Theta-k},
\end{equation}
where $S(k,\mu)=\frac{1}{\mu!}\sum_{j=0}^{\mu}(-1)^{\mu-j}\binom{\mu}{j}j^{k}$
is the Stirling number of the second kind \cite{Stanley2011}. \end{lem}
\begin{IEEEproof}
We divide the ball assignment process into two steps. Firstly we randomly
choose $k\geq\mu$ balls, which has $\binom{\Theta}{k}$ ways, and
put the chosen balls into the $\mu$ boxes such that each box has
at least one balls, which has $\mu!S(k,\mu)$ ways. Secondly we randomly
put the rest $\Theta-k$ balls into the rest $n-\mu$ boxes. To combine
these two steps and sum over $\mu\leq k\leq\Theta$ proves the lemma.
\end{IEEEproof}
When there are $K$ users, by using the two lemmas above, we can count
the number of $\{n_{\delta,i}\}_{i=2}^K$ events in which no three users' offsets can form the feasible
ring as shown in Fig. \ref{fig:3.user.ring}. To ease the derivation,
we assume $N$ is a multiple of $4$, that is, $\frac{N}{4}\in\mathcal{N}=\{1,2,\cdots\}$.

\begin{figure}
\begin{raggedright}
\subfloat[\label{fig:Type-1}Type 1.]{\begin{centering}
\includegraphics[scale=1.5]{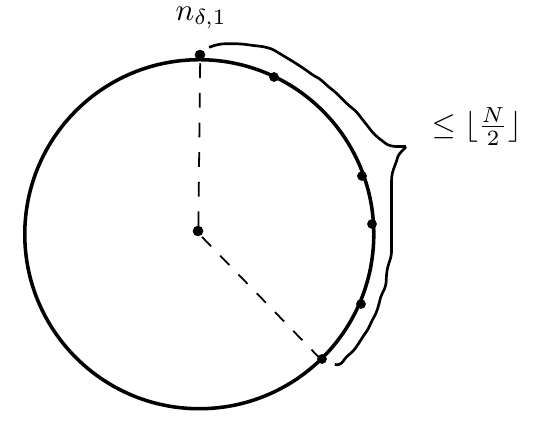}
\par\end{centering}

}\subfloat[\label{fig:Type-2}Type 2]{\centering{}
\includegraphics[scale=1.5]{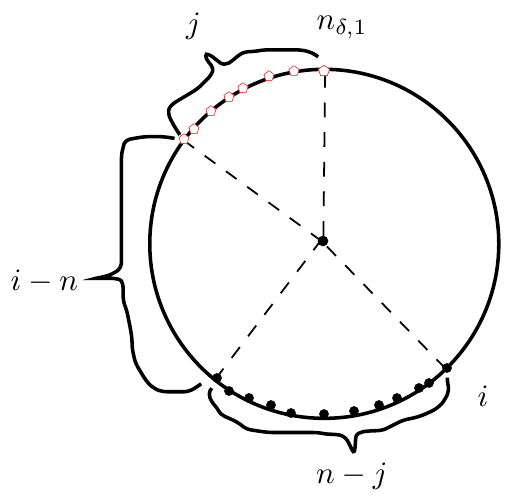}}
\par\end{raggedright}

\caption{\label{fig:thm.ball.box}Events with no BIA-feasible 3-tuple.}
\end{figure}

\begin{thm}
Given $N$ the coherence time subject to $\tfrac{N}{4}$ being an integer, let $f(N, K, 3)$
be the number of $\{n_{\delta,i}\}_{i=2}^K$ events in which no three users, from the $K$ users, can help the transmitter to form a BIA-feasible BC. Then
\begin{eqnarray}
f(N,K,3) & = & 1+\sum_{n=2}^{N/2}\left(2[n^{K-1}-(n-1)^{k-1}]+(n-2)[n^{K-1}-2(n-1)^{K-1}+(n-2)^{K-1}]\right)\nonumber \\
 &  & + (2^{K-1}-1)+\sum_{n=3}^{N/4+1}(n-1)\left[2(n-3)\gamma(n,K-1,3)+3\gamma(n,K-1,2)\right]\nonumber \\
 &  & +\sum_{n=3}^{N/4+1}(n-1)(n-3)\left[\frac{1}{2}(n-4)\gamma(n,K-1,4)+\gamma(n,K-1,3)\right] + \nonumber \\
 &  & \sum_{n=N/4+2}^{N/2}(\frac{N}{2}-n+1)(n-1)\left(2\gamma(n,K-1,3)+\frac{1}{2}(n-4)\gamma(n,K-1,4)\right)\label{eq:no.BIA.events.num}
\end{eqnarray}
where

\begin{subequations}

\begin{eqnarray}
\gamma(n,K-1,2) & = & n^{K-1}-2(n-1)^{K-1}+(n-2)^{K-1},\\
\gamma(n,K-1,3) & = & n^{K-1}-3(n-1)^{K-1}+3(n-2)^{K-1}-(n-3)^{K-1},\\
\gamma(n,K-1,3) & = & n^{K-1}-4(n-1)^{K-1}+6(n-2)^{K-1}-4(n-3)^{K-1}+(n-4)^{K-1}.
\end{eqnarray}

\end{subequations}\end{thm}
\begin{IEEEproof}
We cast this problem into a ball-box problem \cite{Stanley2011}, in which $N$
boxes are labeled from $1$ to $N$ counter-clockwise forming a ring as shown in Fig.~\ref{fig:thm.ball.box}, and a user with offset $n_{\delta,k}$
is put into the box with label $n_{\delta,k}+1$. As before, we set $n_{\delta,1}=0$ as the benchmark. We refer to the length of an arc as the number of
the boxes in the arc. Then, we divide all events
in which no three users can meet the BIA condition into two types. The first is that the
length of the arc which contains all users is no greater
than $\frac{N}{2}$, and the second larger than $\frac{N}{2}$.

\textbf{Type I}), We start with counting the number of the first type events.
Given an arc with the
length $n\leq\frac{N}{2}$ as shown in Fig. \ref{fig:Type-1}, we
count the event number by applying the similar argument developed
in \cite{Zhou2012a}.

If $n_{\delta,1}$ is one end point of the arc,
then the rest $K-1$ users can be randomly loaded into the $n$ boxes
on the arc subject to the condition that the other end of the arc
must be occupied by at least one user. The number of
such events is $2[n^{K-1}-(n-1)^{K-1}]$, where the coefficient $2$ reflects $n_{\delta,1}$
can be either of the two end points. If $n_{\delta,1}$ is not any
end point of the arc, then the location of the arc relative to $n_{\delta,1}$
has $n-2$ possibilities, and for each possibility the rest $K-1$
users can be randomly loaded into the $n$ boxes subject to that the two end points of the arc must be
occupied. The number of such events is
$(n-2)[n^{K-1}-2(n-1)^{K-1}+(n-2)^{K-1}]$, i.e., $(n-2)\gamma(n,K-1,2)$. Combining them, and summing
over $1\leq n\leq\frac{N}{2}$, we get the number of type-I events:
\begin{equation}
f_{1}(N,K,3)=1+\sum_{n=2}^{N/2}\left(2[n^{K-1}-(n-1)^{K-1}]+
(n-2)\gamma(n,K-1,2)\right).
\end{equation}

\textbf{Type II}), Now we count the type II events. As shown in Fig.~\ref{fig:Type-2}, we assume all users are located in two arcs, which occupy $n$
boxes in total; one arc is filled by users denoted by pentagons,
the other by users denoted by dots.

(II,a): In this part, we count the events in which $n_{\delta,1}$ is one end point of the pentagon arc. In such events, as shown in Fig.~\ref{fig:Type-2}, the
position of the end point of the dot arc $i$ should satisfy the condition
$\frac{N}{2}+1\leq i\leq n+\frac{N}{2}-1$, otherwise the event would
belong to the type I above. Given such a $i$, it requires
that $j\leq\frac{N}{4}$ and $n-j\leq\frac{N}{4}$, otherwise a BIA-feasible
3-tuple will arise according to Lemma~\ref{lem:3.user.no.benchmark}. In the following we count the number of such events as $n$ takes three different kinds of values.

Given $n=2$, then $i=\tfrac{N}{2}+1$, the rest $K-1$ balls are randomly put
into two boxes subject to that the other box must be occupied, resulting in $2^{K-1}-1$ such events.

Given $3\leq n\leq\frac{N}{4}+1$ and suitable $i$, the pentagon arc
length $j$ can take three kinds of values. Firstly, when $2\leq j\leq n-2$, the user $n_{\delta,1}$
can be either of the end points of the arc. So the number of possibilities
of the arc combination with $n_{\delta,1}$ being one end point is
$2(n-3)$, and the rest $K-1$ users can be randomly located at the
two arcs with totally $n$ boxes subject to that the other three
end points of the two arcs must be occupied, resulting
in $2(n-3)\gamma(n,K-1,3)$ events. Secondly, when $j=1$,
the two ends of the pentagon arc is the same. Consequently there are $K-1$ users
to be located subject to that the two end points of the dot arc
must be occupied, resulting in $\gamma(n,K-1,2)$ events. Thirdly, when $j=n-1$,
the dot arc has only one end point to fill, resulting in
$2\gamma(n,K-1,2)$ events, in which $2$ reflects $n_{\delta,1}$ can be
either of the two end points of the pentagon arc. Combining the cases of $j=1$,
$2\leq j\leq n-2$, and $j=n-1$, we conclude that
the number of events for the given $n$ and $i$ is $2(n-3)\gamma(n,K-1,3)+3\gamma(n,K-1,2)$.

Given $\frac{N}{4}+2\leq n\leq\frac{N}{2}$ and $i$, the pentagon
arc length can not be $j=1$ or $j=n-1$, but $n-\tfrac{N}{4}\le j\leq\tfrac{N}{4}$.
Referring to the argument above for $j\neq1$ and $j\neq n-1$, we
get $2(\frac{N}{2}-n+1)\gamma(n,K-1,3)$ such events, where
$\frac{N}{2}-n+1$ stands for the possible choices of $j$ based on $n-\frac{N}{4}\le j\leq\frac{N}{4}$.

(II.b): Now we consider the events in which $n_{\delta,1}$ is not one end point of the pentagon arc,
but an internal point of the arc. It is clear that such events can only happen when $j\geq3$, and thus $n\geq 4$.
We first consider the case of $4\leq n\leq\frac{N}{4}+1$. Given such a $n$ and suitable $i$, if $3\leq j\leq n-2$, the
user $n_{\delta,1}$ can take $j-2$ internal positions, resulting
in $\sum_{j=3}^{n-2}(j-2)$ possibilities. The rest $K-1$ users can
be randomly located subject to that each of the end points of the
two arcs must be occupied by at least one user, resulting in $\sum_{j=3}^{n-2}(j-2)\gamma(n,K-1,4)$
events. If $j=n-1$, the dot arc only has one
end point, and the three end points must be occupied, generating $(j-2)\gamma(n,K-1,3)=(n-3)\gamma(n,K-1,3)$
events. Combining $3\leq j\leq n-2$ and $j=n-1$, we
can find $\sum_{j=3}^{n-2}(j-2)\gamma(n,K-1,4)+(n-3)\gamma(n,K-1,3)$ such events,
for $4\leq n\leq\frac{N}{4}+1$ and suitable $i$.

Then we consider the case of $\frac{N}{4}+2\leq n\leq\frac{N}{2}$. Given such $n$,
the length of pentagon arc is constrained by $n-\frac{N}{4}\le j\leq\frac{N}{4}$,
which excludes the chance of $j=n-1$. Then referring to the argument
above for $j\neq n-1$, we get $\sum_{j=n-N/4}^{N/4}(j-2)\gamma(n,K-1,4)$ events.

Finally, we note that no matter whichever type-II case above happens, $i$ can take $n+\frac{N}{2}-1 - (\frac{N}{2}+1)+1 = n-1$ choices (cf.~$\frac{N}{2}+1\leq i\leq n+\frac{N}{2}-1$ for $n_{\delta,1}$ as shown in Fig.~\ref{fig:Type-2}). In summary, combining all cases above for type II, and summing
over the $n-1$ choices of $i$, and over either $2\leq n\leq\frac{N}{4}+1$
or $\frac{N}{4}+2\leq n\leq\frac{N}{2}$, we get the number of type-II events
\begin{eqnarray}
f_{2}(N,K,3) & = & (2^{K-1}-1)+\sum_{n=3}^{N/4+1}\sum_{i=N/2+1}^{n+N/2-1}\left[2(n-3)\gamma(n,K-1,3)+3\gamma(n,K-1,2)\right]\nonumber \\
 &  & +\sum_{n=4}^{N/4+1}\sum_{i=N/2+1}^{n+N/2-1}\left[\sum_{j=3}^{n-2}(j-2)\gamma(n,K-1,4)+(n-3)\gamma(n,K-1,3)\right] + \nonumber \\
 &  & \sum_{n=N/4+2}^{N/2}\sum_{i=N/2+1}^{n+N/2-1}\left(2(\frac{N}{2}-n+1)\gamma(n,K-1,3)+\sum_{j=n-N/4}^{N/4}(j-2)\gamma(n,K-1,4)\right)\nonumber \\
\end{eqnarray}

In summary, the number of events which contains no BIA-feasible 3-tuple
is obtained by adding $f_{1}(N,K,3)$ and $f_{2}(N,K,3)$, which proves the
theorem.
\end{IEEEproof}

By using the theorem above, we can prove the following theorem.
\begin{thm}
Given a homogeneous $2\times 1$ broadcast network with $K$ users and $N$ the coherence time subject to $\tfrac{N}{4}$ being an integer, let $P(N,K,3)$ be the probability that the
transmitter finds, among the $K$ users, three users to form a BIA-feasible $3$-user $2\times 1$ BC. Then
\begin{equation}\label{eq:P.N.K}
P(N,K,3)=1-\frac{f(N,K,3)}{N^{K-1}}.
\end{equation}
\end{thm}
\begin{proof}
Based on the fact that the number of possible $\{n_{\delta,i}\}_{i=2}^K$ events is $N^{K-1}$ and $n_{\delta,i}$ is uniformly distributed for all $2\leq i \leq K$, the proof of the theorem is quite straightforward.
\end{proof}

\section{Results and discussions}

\begin{figure}
\begin{centering}
\includegraphics[scale=0.9]{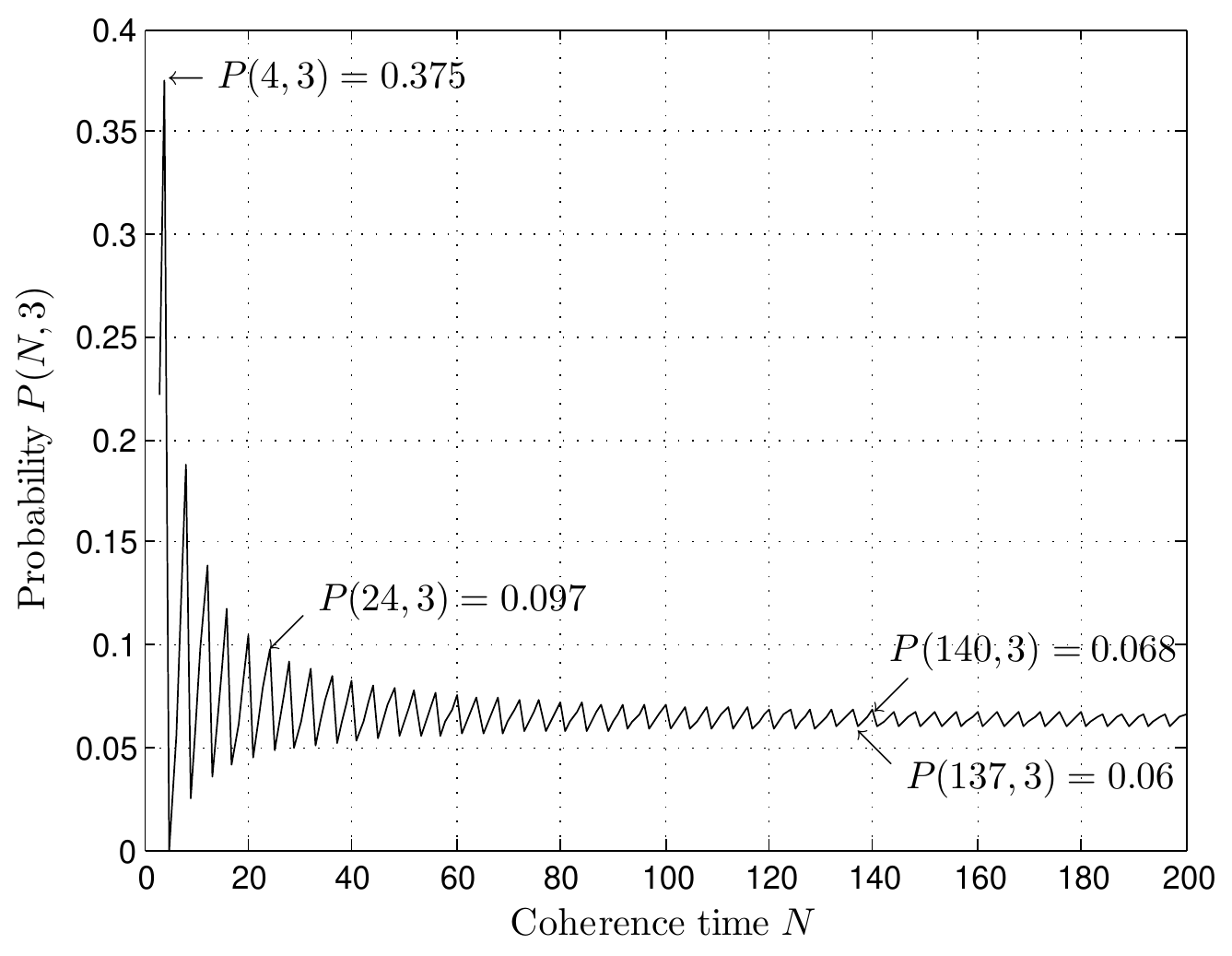}
\par\end{centering}
\caption{Probability $P(N,3)$ as a function of coherence time $N$.}
\label{fig:Ncoherence}
\end{figure}

\begin{figure}
\begin{centering}
\includegraphics[scale=0.9]{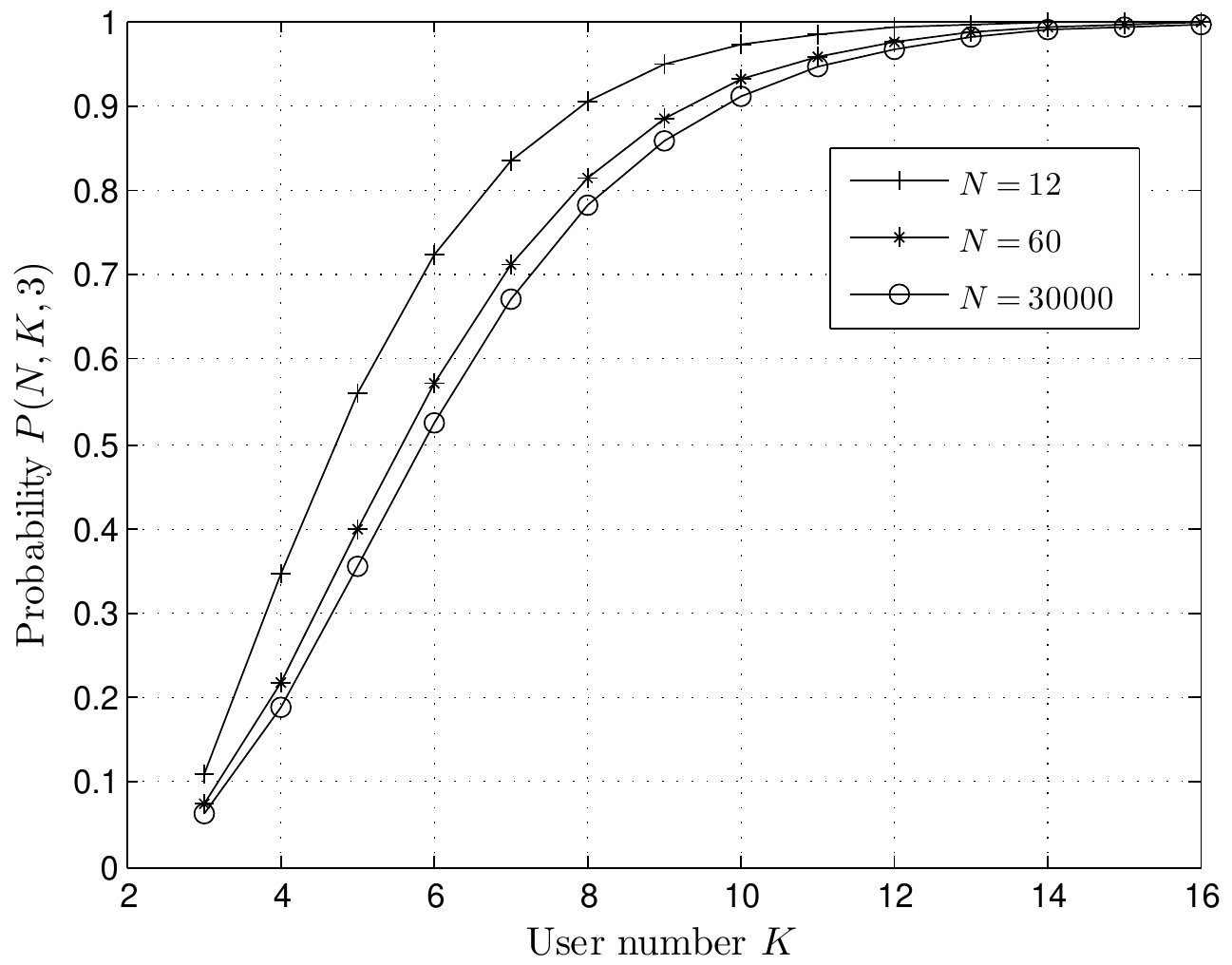}
\par\end{centering}
\caption{Probability $P(N,K,3)$ as a function of user number $K$ for different coherence time $N$.}
\label{fig:Kuser}
\end{figure}

Given a homogeneous $3$-user $2\times 1$ BC, Fig.~\ref{fig:Ncoherence} shows its BIA-feasible probability $P(N,3)$ as a function of coherence time $N$. When $N=4$, the BC has the maximal BIA-feasible probability $P(4,3)=0.375$. Rather than a monotonic function of $N$, the probability $P(N,3)$ forms a Cauchy sequence over $N$ and converges to a constant value when $N$ goes large. In specific, the probability limit $P(\infty,3)$ falls in $(P(137,3),P(140,3))=(0.06,0.068)$.

Although a homogeneous $3$-user $2\times 1$ BC has a small chance to be BIA-feasible, as presented above, it is expected that the probability increases if the three users can be selected favorably from a number of independent homogeneous users. Quantitatively, Fig.~\ref{fig:Kuser} shows the probability $P(N,K,3)$ as a function of user number $K$ for $N=12,60,3000$. It is observed that, when fixing $N$,
the probability of finding 3 users among $K$ users to form a BIA-feasible BC
increases monotonically with $K$. When $N\geq 60$, the $P(N,K,3)$ curve is close to the $P(\infty,K,3)$ curve, which is fairly portrayed by $P(30000,K,3)$. Based on the $P(30000,K,3)$ curve, we observe that the probability is larger than $95\%$ for one to find three users from the $K$ users to form a BIA-feasible $3$-user BC when $K=11$. The probability is almost one when $K\geq 15$.

\section{Conclusion}

In summary, this paper shows a sufficient condition for a homogeneous $3$-user $2\times 1$ BC to be BIA-feasible, that is, to achieve the optimal $\frac{3}{2}$ DoF by using interference alignment without the need of CSIT. Considering a homogeneous $2\times 1$ broadcast network in which the coherence blocks of $K\geq 11$ users are randomly offset with uniform distribution, there is more than 95\% chance that the two-antenna transmitter can find three users to form a BIA-feasible $3$-user BC. The results in this paper are also applicable to the $2\times 3$ X channel since the $3$-user $2\times 1$ BC has the same system model as the $2\times 3$ X channel, in which two single-antenna transmitters broadcast to three single-antenna receivers.

\end{document}